\documentclass{article}
\usepackage[T1]{fontenc}
\usepackage{enumerate}
\usepackage{enumitem}
\usepackage{graphicx}
\usepackage{amsmath}
\usepackage{amssymb}
\usepackage{amsthm}
\usepackage[style=trad-abbrv,doi=false,url=false,isbn=false,date=year]{biblatex}

\bibliography{PostorderBib}

\theoremstyle{plain}
\newtheorem{theorem}{Theorem}
\newtheorem{lemma}{Lemma}
\newtheorem{conjecture}{Conjecture}

\theoremstyle{definition}
\newtheorem{remark}{Remark}

\newenvironment{thmenum}
 {\begin{enumerate}[label=\upshape(\alph*),ref=\thetheorem(\alph*)]}
 {\end{enumerate}}

\DeclareMathOperator{\Root}{root}
\DeclareMathOperator{\Left}{left}
\DeclareMathOperator{\Right}{right}
\DeclareMathOperator{\Parent}{parent}
\DeclareMathOperator{\Postorder}{postorder}
\DeclareMathOperator{\Preorder}{preorder}

\DeclareMathOperator{\BST}{BST}  %
\DeclareMathOperator{\DF}{DF}  %
\DeclareMathOperator{\BB}{BB}  %
\DeclareMathOperator{\LCA}{lca}  %

\title{Splaying Preorders and Postorders\footnote{Research at
Princeton University partially supported by an innovation research grant from
Princeton and a gift from Microsoft.}}

\author{
  Caleb C. Levy\footnote{Baskin School of Engineering, UC Santa Cruz; cclevy@ucsc.edu.}\\
  \and
  Robert E. Tarjan\footnote{Department of Computer Science, Princeton University, and Intertrust Technologies; ret@cs.princeton.edu.}
}

\begin{document}
\maketitle              %
\begin{abstract}
Let $T$ be a binary search tree of $n$ nodes with root $r$, left subtree
$L=\Left(r)$, and right subtree $R=\Right(r)$. The \emph{preorder} and
\emph{postorder} of $T$ are defined as follows: the preorder and postorder of
the empty tree is the empty sequence, and
\begin{align*}
\Preorder(T) &= (r)\oplus\Preorder(L)\oplus\Preorder(R)\\
\Postorder(T) &= \Postorder(L)\oplus\Postorder(R)\oplus(r),
\end{align*}
where $\oplus$ denotes sequence concatenation.\footnote{We will refer to any
such sequence as \emph{a preorder} or \emph{a postorder}.} We prove the
following results about the behavior of splaying \cite{SPLAY_JOURNAL} preorders
and postorders:

\begin{enumerate}
\item\label{item:LinearTimePreorder} Inserting the nodes of $\Preorder(T)$ into
an empty tree via splaying costs $O(n)$.
(Theorem~\ref{thm:InsertionSplayingPreorders}.)
\item\label{item:LinearTimePostorder} Inserting the nodes of $\Postorder(T)$
into an empty tree via splaying costs $O(n)$.
(Theorem~\ref{thm:InsertionSplayingPostorders}.)
\item\label{item:LinearTimeBalancedTrees} If $T'$ has the same keys as $T$ and
$T$ is \emph{weight-balanced} \cite{WEIGHT_BALANCED_TREES} then
splaying either $\Preorder(T)$ or $\Postorder(T)$ starting from $T'$ costs
$O(n)$. (Theorem~\ref{thm:BalancedTraversal}.)
\end{enumerate}

For \ref{item:LinearTimePreorder} and \ref{item:LinearTimePostorder}, we use
the fact that preorders and postorders are \emph{pattern-avoiding}: i.e. they
contain no subsequences that are order-isomorphic to $(2,3,1)$ and $(3,1,2)$,
respectively. Pattern-avoidance implies certain constraints on the manner in
which items are inserted. We exploit this structure with a simple potential
function that counts inserted nodes lying on access paths to uninserted nodes.
Our methods can likely be extended to permutations that avoid more general
patterns. The proof of~\ref{item:LinearTimeBalancedTrees} uses the fact that
preorders and postorders of balanced search trees do not contain many large
``jumps'' in symmetric order, and exploits this fact using the dynamic finger
theorem \cite{DYNAMIC_FINGER_1,DYNAMIC_FINGER_2}.

Items \ref{item:LinearTimePostorder} and \ref{item:LinearTimeBalancedTrees} are
both novel. Item \ref{item:LinearTimePreorder} was originally proved by
Chaudhuri and H\"oft \cite{SPLAY_PREORDER}; our proof simplifies theirs. These
results provide further evidence in favor of the elusive \emph{dynamic
optimality conjecture} \cite{SPLAY_JOURNAL}.
\end{abstract}

\paragraph{Outline.} Section \ref{sec:Introduction} discusses the mathematical
preliminaries, historical background, and context for this investigation, and
Section \ref{sec:RelatedWork} samples some related work. Familiar readers may
skip directly to the main results and their proofs, in Sections
\ref{sec:MainResults} and \ref{sec:BalancedTrees}. Section
\ref{sec:MainResults} proves that inserting both preorders and postorders via
splaying takes linear time. Section \ref{sec:BalancedTrees} establishes that
splaying preorders and postorders of \emph{balanced} search trees
\cite{WEIGHT_BALANCED_TREES} takes linear time, regardless of starting tree.
Section \ref{sec:GeneralPatterns} provides our thoughts on how to analyze
insertion splaying permutations that avoid more general patterns, particularly the class of ``$k$-increasing'' sequences \cite{PATTERN_AVOIDANCE}.

\section{Preliminaries}\label{sec:Introduction}

\subsubsection*{Binary Search Trees}

A \emph{binary tree} $T$ contains of a finite set of \emph{nodes}, with one
node designated to be the \emph{root}. All nodes have a \emph{left} and a
\emph{right} \emph{child} pointer, each leading to a different node. Either or
both children may be \emph{missing}, and we denote a missing child by
\texttt{null}. Every node in $T$, save for the root, has a single
\emph{parent} node of which it is a child. (The root has no parent.) The
\emph{size} of $T$ is the number of nodes it contains, and is denoted $|T|$.

There is a unique path from $\Root(T)$ to every other node $x$ in $T$, called
the \emph{access path for $x$ in $T$}. If $x$ is on the access path for $y$
then we say $x$ is an \emph{ancestor} of $y$, and $y$ is a \emph{descendent} of
$x$. We refer to the subtree comprising $x$ and all of its descendants as the
subtree rooted at $x$. Nodes thus have \emph{left} and \emph{right subtrees}
rooted respectively at their left and right children. (Subtrees are
\emph{empty} for \texttt{null} children.) The \emph{depth} of the node $x$,
denoted $d_T(x)$, is the number of edges on its access path. Its
\emph{right-depth} is the number of right pointers followed, and its
\emph{left-depth} is the number of left pointers followed.

In a \emph{binary search tree}, every node has a unique \emph{key}, and the
tree satisfies the \emph{symmetric order} condition: every node's key is
greater than those in its left subtree and smaller than those in its right
subtree. The binary search tree derives its name from how its structure enables
finding keys. To find a key $k$ initialize the current node to be the root.
While the current node is not \texttt{null} and does not contain the given key,
replace the current node by its left or right child depending on whether $k$ is
smaller or larger than the key in the current node, respectively. The search
returns the last current node, which contains $k$ if $k$ is in the tree and
otherwise \texttt{null}.

The \emph{lowest common ancestor} of $x$ and $y$ in $T$, denoted $\LCA_T(x,y)$,
is the deepest node shared by the access paths of both $x$ and $y$. Since the
root is a common ancestor of any pair of nodes in $T$ and $T$ is finite,
$\LCA_T(x,y)$ exists and is well defined. Furthermore
$\min\{x,y\}\le\LCA_T(x,y)\le\max\{x,y\}$.

To \emph{insert} a new key $k$ into a binary search tree $T$, we first do a
search for $k$ in $T$. When the search reaches a missing node, we replace this
node with a node containing the key $k$. (Inserting into an empty tree makes
$k$ the root key.)

\subsubsection*{Rotation}

Binary search trees are the canonical data structure for maintaining an ordered
set of elements, and are building blocks in countless algorithms. Perhaps the
most attractive feature of binary search trees is that the number of
comparisons required to find an item in an $n$-node binary search tree is
$O(\log n)$, provided that that the tree is properly arranged, which is good in
theory and practice. However, without exercising care when inserting nodes, a
binary search tree can easily become unbalanced (for example when inserting
$1,2,\dots,n$ in order), leading to search costs as high as $\Omega(n)$. Thus,
binary search trees require some form of maintenance and restructuring for good
performance.

\begin{figure}
\centering
\includegraphics[width=\columnwidth,trim={.15in 1.3in .25in
2in},clip]{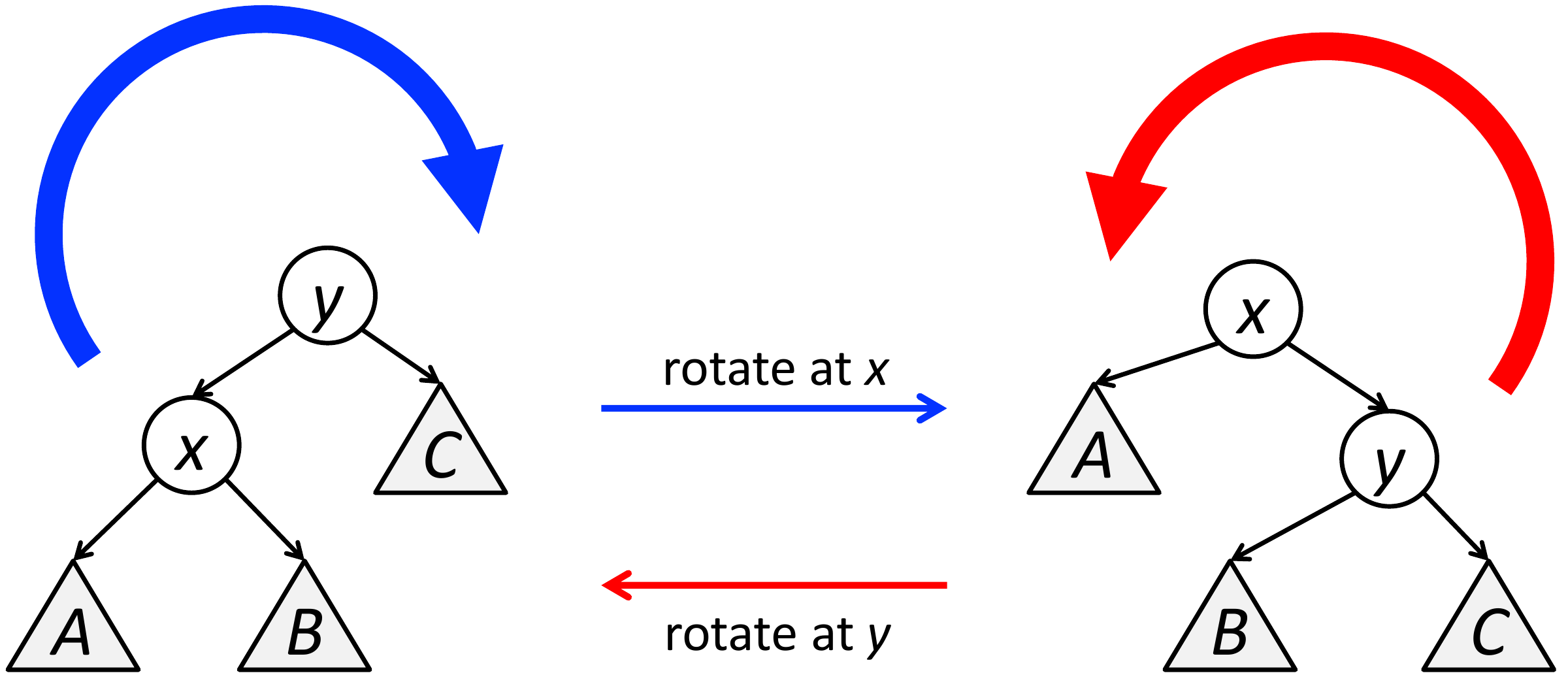}
\caption{Rotation at node $x$ with parent $y$, and reversing the effect
by rotating at $y$.}\label{fig:Rotation}
\end{figure}

We will employ a restructuring primitive called \emph{rotation}. A rotation at
left child $x$ with parent $y$ makes $y$ the right child of $x$ while
preserving symmetric order. A rotation at a right child is symmetric, and
rotation at the root is undefined. (See Figure \ref{fig:Rotation}). A rotation
changes three child pointers in the tree.

Rotations were first employed in ``balanced'' search trees, which include AVL
trees \cite{AVL_TREES}, Red-Black trees \cite{RED_BLACK_TREES}, weight-balanced
trees \cite{WEIGHT_BALANCED_TREES}, and more recently weak AVL trees
\cite{WEAK_AVL_TREES}. These trees augment nodes with bits that provide rough
information about how ``balanced'' each node's subtree is. Whenever an item is
inserted or deleted, rotations are performed to restore invariants on the
balance bits that ensure all search paths have $O(\log n)$ nodes. While
balanced searched trees are not the focus of this work, they were progenitors
for the main algorithm of interest.

\subsubsection*{Splay}

The Splay algorithm \cite{SPLAY_JOURNAL} eschews keeping track of balance
information, replacing it with an intriguing notion: instead of adjusting the
search tree only after insertion and deletion, Splay modifies the tree after
\emph{every} search.

\begin{figure}
\centering
\includegraphics[width=\columnwidth]{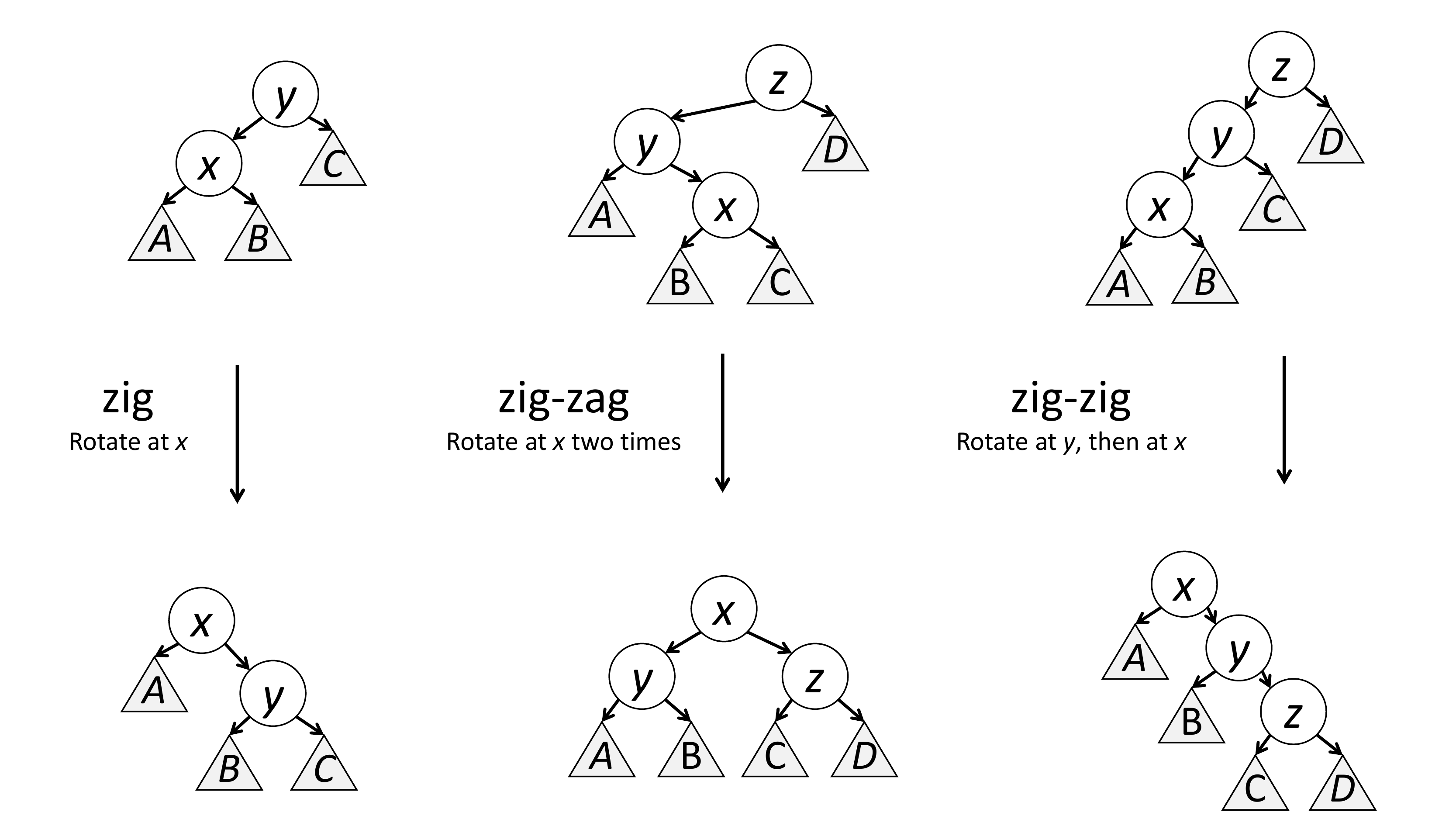}
\caption{A splaying step at node $x$. Symmetric variants not shown. Triangles
denote subtrees.}\label{fig:SplaySteps}
\end{figure}

The algorithm begins with a binary search for a key in the tree. Let $x$ be the
node returned by this search. If $x$ is not \texttt{null} then the algorithm
repeatedly applies a ``splay step'' until $x$ becomes the root. A splay step
applies a certain series of rotations based on the relationship between $x$,
its parent, and its grandparent, as follows. If $x$ has no grandparent (i.e.
$x$'s parent is the root), then rotate at $x$ (this case is always terminal).
Otherwise, if $x$ is a \emph{left} child and its parent is a \emph{right}
child, or vice-versa, rotate at $x$ twice. Otherwise, rotate at $x$'s parent,
and then rotate at $x$. Sleator and Tarjan \cite{SPLAY_JOURNAL} assigned the
respective names \emph{zig}, \emph{zig-zag} and \emph{zig-zig} to these three
cases. The series of splay steps that bring $x$ to the root are collectively
called to as \emph{splaying at $x$}, or simply \emph{splaying $x$}. The three
cases are depicted in Figure
\ref{fig:SplaySteps}.

The \emph{cost} of splaying a single item $x$ in $T$ is defined to be
$d_T(x)+1$.\footnote{We absorb the search cost into the rotations.} If
$X=(x_1,\dots,x_m)$ is a sequence of requested keys in $T$ then the cost of
splaying $X$ starting from $T$ is defined as $m+\sum_{i=1}^m d_{T_{i-1}}(x_i)$,
where $T_0=T$, and for $1\le i \le m$, we form $T_i$ by splaying $x_i$ in
$T_{i-1}$. To perform \emph{insertion splaying}, insert a key into the tree and
then splay the newly created node. The cost of an insertion splay is the cost
splaying the new node.

While splaying individual items can cost $\Omega(n)$, the \emph{total} cost of
splaying $m$ requested items in a tree of size $n>0$ is $O((m+n)\log n)$.
Hence, the worst case cost of a splay operation, \emph{amortized} over all the
requests, is the same as any balanced binary search tree. This is perhaps
surprising for an algorithm that keeps no record of balance information.

What makes Splay truly remarkable is how it takes advantage of ``latent
structure'' in the request sequence, and provides more than simple
``worst-case'' guarantees. As just one example, if $t_X(i)$ is the number of
different items accessed before access $i$ since the last access to item $x_i$
(or since the beginning of the sequence if $i$ is the first access to $x_i$),
then the cost to splay $X$ starting from $T$ is $O(n\log n
+\sum_{j=1}^m\log(t_X(j)+1))$ \cite{SPLAY_JOURNAL}.\footnote{Note that $O(\log
n)$ amortized cost per splay is a corollary of this.} (This is called the
``working set'' property.) Thus, Splay exploits ``temporal locality'' in the
access pattern.

Splay simultaneously exploits ``spatial'' locality, as shown by the following
theorem (originally conjectured in \cite{SPLAY_JOURNAL}) that we will use later
on:
\begin{theorem}[Dynamic Finger \cite{DYNAMIC_FINGER_1,DYNAMIC_FINGER_2}]
\label{thm:DynamicFinger}
Let the \emph{rank} of $x$ in $T$, denoted $r_T(x)$, be the number of nodes in
$T$ whose keys are less than or equal the key in $x$. The cost of splaying
$X=(x_1,\dots,x_m)$ starting from $T$ is $O(|T|+m+\DF_T(X))$, where
$\DF_T(X)\equiv\sum_{i=2}^m \log_2(|r_T(x_i)-r_T(x_{i-1})|+1)$.
\end{theorem}

In fact, the properties of Splay inspired the authors of \cite{SPLAY_JOURNAL}
to speculate on a much stronger possibility: that Splay's cost is always within
a constant factor of the ``optimal'' way of executing the requests. Formally,
an \emph{execution} $E$ for $(X,T)$ comprises the following. Let $T_0=T$, and
for $1\le i\le m$, we perform some number $e_i\ge 0$ of rotations starting from
$T_{i-1}$ to form $T_i$, followed by a search for $x_i$. The \emph{cost} of
this execution is $\sum_{i=1}^m(1+e_i+d_{T_{i-1}}(x_i))$. The \emph{optimal
cost} $\operatorname{OPT}(X,T)\equiv\min\{\operatorname{cost}(E)\mid\text{$E$
executes $(X,T)$}\}$. The following conjecture has spawned a great deal of
related research (see \S\ref{sec:RelatedWork}):

\begin{conjecture}[Dynamic Optimality \cite{SPLAY_JOURNAL}]
$\operatorname{cost}_{\text{splay}}(X,T)=O(\operatorname{OPT}(X,T))$.
\end{conjecture}

The conjecture remains open. In fact, there is no sub-exponential time
algorithm \emph{whatsoever} that is known to compute, even to within a constant
factor, the cost of an optimum binary search tree execution for an instance.
There are several known lower bounds \cite{GEOMETRY,WILBER}, none known to be
tight (though some conjectured to be).

\subsubsection*{Pattern-Avoidance}

For simplicity, we restrict subsequent discussion to \emph{permutation} request
sequences (i.e. no key is requested twice). By \cite{PATTERN_AVOIDANCE}, any
algorithm that achieves optimal cost on all permutations can be extended to an
algorithm that is optimal for all request sequences.

An auxiliary question to determining if Splay (or any other algorithm) is
dynamically optimal is: ``what class(es) of permutations have optimum
executions with `low' cost?'' This issue is not a mere curiosity, as almost
every permutation of length $n$ has optimal execution cost $\Theta(n\log n)$
\cite{BST_EXECUTION_ENCODINGS}, a bound achieved by any balanced search tree.
Thus, in the absence of insertions or deletions, adjusting the tree after every
access only gives an advantage on a small subset of ``structured'' request
sequences. In addition, these structured request sequences provide candidate
counter-examples to dynamic optimality. In this work, we focus on certain
\emph{pattern-avoiding} permutations: those that do not contain any
subsequences of a specified type. More formally:\footnote{The following
definitions and theorems are taken from \cite[Chapter 1.3]{KOZMA_THESIS},
almost verbatim.}

 Two permutations $\alpha=(a_1,\dots, a_n)$ and
$\beta=(b_1,\dots, b_n)$ of the same length are \emph{order-isomorphic}
if their entries have the same relative order, i.e. $a_i<a_j\iff
b_i<b_j$. For example, $(5,8,1)$ is order-isomorphic to $(2,3,1)$. A
sequence $\pi$ \emph{avoids} a sequence $\alpha$ (or is called
\emph{$\alpha$-avoiding}) if it has no subsequence that is order-isomorphic
with $\alpha$. If $\pi$ is $\alpha$-avoiding then all subsequences of $\pi$ are
$\alpha$-avoiding. We use $\pi\setminus\alpha$ as shorthand for ``an
(arbitrary) permutation $\pi$ that avoids $\alpha$.'' Both preorders and
postorders may be characterized as pattern-avoiding permutations:

\begin{lemma}[Lemma 1.4 from \cite{KOZMA_THESIS}] For any permutation $\pi$: \item
\begin{thmenum}
\item\label{lem:Preorders231} $\pi=\Preorder(T)$ for some binary search tree $T$ $\iff\pi$ avoids $(2,3,1)$.
\item\label{lem:Postorders312} $\pi=\Postorder(T)$ for some binary search tree $T$ $\iff\pi$ avoids $(3,1,2)$.
\end{thmenum}
\end{lemma}
\begin{proof}[Sketch]
For preorders, Kozma builds a bijection between binary search trees and
$(2,3,1)$-avoiding sequences, and uses a simple argument by contradiction to
show preorders avoid $(2,3,1)$ \cite{KOZMA_THESIS}. The proof for postorders is
a nearly symmetric variation of this argument.
\end{proof}

\section{Related Work}\label{sec:RelatedWork}

The first result about Splay's behavior on pattern-avoiding request sequences
was the \emph{sequential access} theorem \cite{SEQUENTIAL_ACCESS}: the cost of
splaying the nodes of $T$ in order is $O(|T|)$. This is a special case of a
corollary\footnote{A priori, the traversal conjecture follows from dynamic
optimality \emph{conditioned} on Splay being optimal with low ``additive
overhead.'' The authors recently proved that this corollary is actually
\emph{unconditional} \cite{A_NEW_PATH}.} of dynamic optimality:
\begin{conjecture}[Traversal \cite{SPLAY_JOURNAL}]
There exists $c>0$ for which the cost of splaying $\Preorder(T)$ starting from
$T'$ is at most $c|T|$ for all pairs of binary search trees $T, T'$ with the
same keys.
\end{conjecture}
Theorem~\ref{thm:InsertionSplayingPreorders} and \cite{SPLAY_PREORDER} is
another special case, when $T=T'$. In \S\ref{sec:MainResults} we prove a new
special case: when $T$ is $\alpha$-weight balanced.

Interest in the behavior of binary search tree algorithms on ``structured''
request sequences was revived by Seth Pettie's analysis of the performance of
Splay-based deque data structures using Davenport-Schinzel sequences
\cite{DAVENPORT_SCHINZEL}, and his later reproof of the sequential access
theorem via the theory of forbidden submatrices \cite{FORBIDDEN_SUBMATRIX}.

This analysis was later adapted to and greatly extended for another binary
search tree algorithm, colloquially known as ``Greedy,'' that was first
proposed as an \emph{off}-line algorithm independently by Lucas
\cite{CANONICAL_FORMS} and Munro \cite{MUNRO_GREEDY}. Greedy is widely
conjectured to be dynamically optimal, and is known to have many of the same
properties of Splay, including the working set \cite{GREEDY_ACCESS_LEMMA} and
dynamic finger \cite{GREEDY_DYNAMIC_FINGER} bounds.

Greedy was later recast as an on-line algorithm in a ``geometric'' view of
binary search trees \cite{GEOMETRY}. This geometric view of Greedy is
especially amenable to forbidden submatrix analysis. In
\cite{PATTERN_AVOIDANCE}, Chalermsook et. al. show that Greedy has
nearly-optimal run-time on a broad class of pattern-avoiding permutations.
Moreover, they demonstrate that if Greedy is optimal on a certain class of
``non-decomposable'' permutations then it is dynamically optimal. Chalermsook
et al.'s analysis was later simplified in \cite{SIMPLER_PATTERN_AVOIDANCE}.

\section{Insertion Splaying Preorders and Postorders}\label{sec:MainResults}

If $\pi=(p_1,\dots,p_n)$ is a permutation then the \emph{insertion tree for
$\pi$}, denoted $\BST(\pi)$, is the binary search tree obtained by starting
from an empty tree and inserting keys in order of their first appearance in
$\pi$.
\begin{lemma}\label{lem:AncestorPrecedence}
If $x$ is a proper ancestor of $y$ in $\BST(\pi)$ then $x$ precedes $y$ in
$\pi$.
\end{lemma}

\begin{proof}
Let $\pi_{\prec y}$ denote the prefix of $\pi$ containing the elements preceding
$y$. By construction, $y$ is inserted as a child of some node $z$ in
$\BST(\pi_{\prec y})$. Every proper ancestor of $y$ is an ancestor of
$z$, thus $x\in\BST(\pi_{\prec y})$. Hence, $x$ precedes $y$.
\end{proof}

Insertion splaying $\pi$ has the same cost as splaying $\pi$ starting from
$\BST(\pi)$.\footnote{This is because the manner in which Splay restructures
the access path is independent of nodes outside the path.} For the purposes of
analysis we will assume that, initially, every node in $\BST(\pi)$ is marked as
\emph{untouched}. An insertion splay marks the node as \emph{touched}, and then
splays the node. The touched nodes form a connected subtree containing the
root, called the \emph{touched subtree}. The untouched nodes form subtrees each
of which contains no touched node. Call an untouched node with a touched parent
a \emph{sub-root}. The subtrees rooted at sub-roots have identical structure in
both the splayed tree and $\BST(\pi)$. By Lemma \ref{lem:AncestorPrecedence},
the next node to be touched is always a sub-root.

For $1\le i\le n$, form $T_i$ by touching and then splaying $p_i$ in $T_{i-1}$,
where $T_0=\BST(\pi)$ starts with all nodes untouched. At any time we define
the \emph{potential} to be the twice the number of touched nodes that are
ancestors of sub-roots, and we define $\Phi_i$ to be the potential of $T_i$.
The \emph{amortized} cost of splaying $p_i$ in $T_{i-1}$ is defined as
$c_i=t_i+\Phi_i-\Phi_{i-1}$, where $t_i$ denotes the actual cost. By a standard
telescoping sum argument, the cost of insertion splaying $\pi$ is $\sum_{i=1}^n
t_i = \sum_{i=1}^n c_i + \Phi_0 - \Phi_n$ \cite{AMORTIZED_TIME}. Since
$\Phi_0=\Phi_n=0$, an upper bound on amortized cost provides an upper bound on
the actual cost.

Pattern-avoidance provides certain information about both $\BST(\pi)$ and about
which sub-root can be touched next. We exploit this information in the next two
sections.

\subsubsection*{Preorders}

There are no restrictions on the possible structure of preorder insertion trees
as $\BST(\Preorder(T))=T$.\footnote{In fact, this property is shared by
\emph{any} permutation $\pi$ for which every node in $T$ appears in $\pi$
before those in its left and right subtrees.} However, the manner in which
sub-roots are chosen is particularly simple.

\begin{lemma}\label{lem:PreorderSubroots}
If $\pi\setminus(2,3,1)=(p_1,\dots,p_n)$ is a preorder then, for $1\le i\le n$,
$p_i$ is the smallest sub-root of $T_{i-1}$, where all nodes begin untouched in
$T_0=\BST(\pi)$ and $T_i$ is formed by touching and splaying $p_i$ in $T_{i-1}$.
\end{lemma}

\begin{proof}
The statement is vacuously true for $i=1$. We prove for $i>1$ by contradiction,
as follows. Suppose $T_{i-1}$ has some sub-root $q$ that is smaller than $p_i$.
Since $q$ and $p_i$ are \emph{both} sub-roots in $T_{i-1}$, they are both
children of respective (though not necessarily distinct) nodes $a$ and $b$ in
$T_{i-1}$. Let $r=\LCA_{T_{i-1}}(a,b)$. Since $q\ne a$ and $p_i\ne b$, all of
$p_i$, $q$ and $r$ are distinct nodes in $T_{i-1}$, and furthermore $q<r<p_i$.
By Lemma \ref{lem:AncestorPrecedence}, $r$ precedes both $q$ and $p_i$ in
$\pi$, and by construction $p_i$ precedes $q$. We thus have $(r,p_i,q)$ is a
subsequence of $\pi$. But $(r,p_i,q)$ is order-isomorphic with $(2,3,1)$,
contradicting $\pi\setminus(2,3,1)$.
\end{proof}

\begin{figure}
\centering
\includegraphics[width=.85\textwidth]{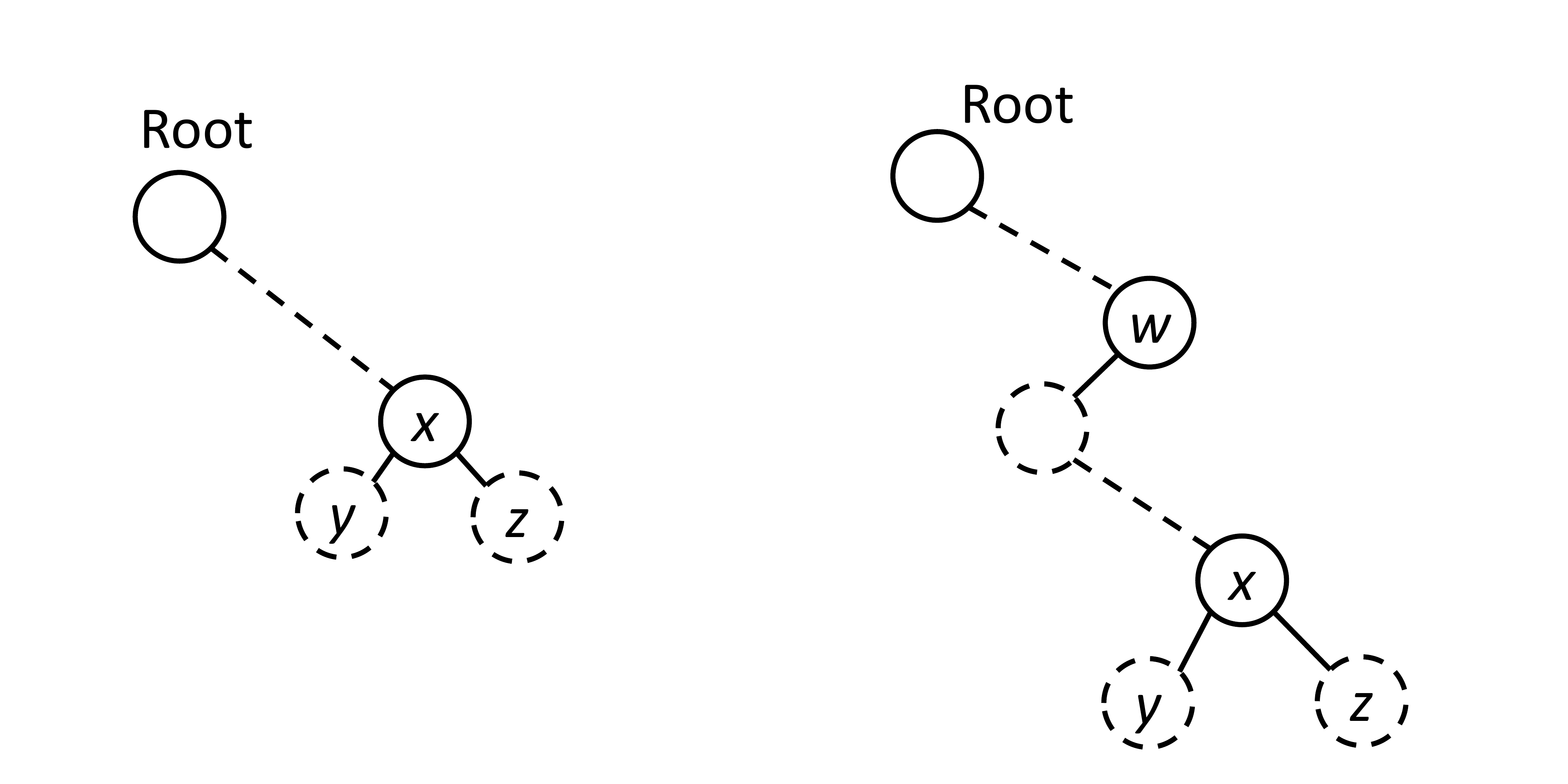}
\caption{Possible locations for the next sub-root $x$ to be insertion splayed
in $\pi\setminus(2,3,1)$. The case on the left occurs when the next splayed
node has left-depth $0$, and the case on the right occurs when it has
left-depth $1$. Dashed nodes may or may not be present, and any number of nodes
may lie on the paths denoted by dashed lines.}\label{fig:PreorderSubroots}
\end{figure}

\begin{theorem}\label{thm:InsertionSplayingPreorders}
Insertion splaying $\Preorder(T)$ keeps each sub-root at left-depth at most $1$
and takes $O(1)$ amortized time per splay operation.
\end{theorem}

\begin{proof}
The theorem is trivial for the first insertion splay. The inductive hypothesis
is that every sub-root has left depth $0$ or $1$. Let $x$ be the next sub-root
to be splayed, and let $y$ and $z$ (either or both of which can be missing) be
its left and right children. Touching $x$ makes $y$ and $z$ into sub-roots.

Suppose $x$ has left depth $0$ before it is touched. Converting $x$ from
untouched to touched (without splaying it) increases the potential by at most
$2$ and gives the new sub-roots $y$ and $z$ left depths of $1$ and $0$,
respectively. (In this case they are the only two sub-roots.) Each splay step,
except possibly the last, is a zig-zig in which $x$ starts as a left child with
parent $p$ and grandparent $g$. After completing the zig-zig, $g$ is no longer
an ancestor of any untouched node, which decreases the potential by $2$. The
zig-zig also preserves the left depths of $y$ and $z$. ($y$ becomes the right
child of $p$.) No other sub-roots can increase left-depth, as $x$ is the
smallest sub-root. If the last splay step is a zig, the potential does not
change (although the length of the path to $y$ increases by $1$).

More complicated is the case in which $x$ has left depth $1$. Converting $x$
from untouched to touched (without splaying it) makes $y$ a sub-root of left
depth $2$ and $z$ a sub-root of left depth $1$. Let $w$ be the parent of the
ancestor of $x$ that is a left child. All other sub-roots are in the right
subtree of $w$, which is unaffected by splaying $x$. The splay of $x$ consists
of $0$ or more left zig-zigs, followed by a zig-zag (which can either
left-right or right-left), followed by zero of more left zig-zigs, followed
possibly by a zig. Each zig-zig reduces the potential by $2$ and preserves the
left depths of all sub-roots. The zig-zag does not increase the potential,
reduces the left depth of $y$ from $2$ to $1$, and that of $x$ from $1$ to $0$,
and preserves the left depth of $z$. Now $x$ has left depth $0$, and the
argument above applies to the remaining splay steps.

By Lemma \ref{lem:PreorderSubroots}, the next node to be splayed will be $y$
if present, otherwise $z$ if present, otherwise $w$ if present. All three of
these items have left-depth $0$ or $1$, hence an identical form to Figure
\ref{fig:PreorderSubroots}. Thus the hypothesis holds.

To obtain the constant factor, we observe that converting $x$ from untouched
to touched increases the potential by $2$. Each zig-zig step pays for itself:
it requires $2$ rotations, paid for by the potential decreasing by at least
$2$. The zig-zag requires $2$ rotations, and the zig requires $1$ rotation. If
the cost of a splay is the number of nodes on the splay path, equal to the
number of rotations plus $1$, we have an amortized cost of $6$ per splay.
\end{proof}

\subsubsection*{Postorders}

Postorder insertion trees are more restricted. A binary search tree $C$ is a
\emph{(left-toothed) comb} if the access path for $x\in C$ always comprises
some number $j\ge0$ of right children followed by some number $k\ge0$ of left
children. The nodes of $C$ are partitioned into \emph{teeth}, where every node
in the $i\textsuperscript{th}$ tooth has right-depth $i-1$. The shallowest node
in a tooth is called the \emph{head}. The insertion trees of postorders are
combs:

\begin{lemma}\label{lem:PostorderInsertionTrees}
If $\pi$ is a postorder then no left child in $\BST(\pi)$ has a right child.
\end{lemma}
\begin{proof}
By contradiction. Let $y$ be a left child in $\BST(\pi)$ with right child $z$,
and let $x=\Parent(y)$. As $z$ is $y$'s right child, $y<z$. Similarly, as both
$y$ and $z$ are in $x$'s left subtree, $y<z<x$. By Lemma
\ref{lem:AncestorPrecedence}, $y$ can be an ancestor of $z$ only if $y$
precedes $z$ in $\pi$, and similarly $x$ must precede $y$. Thus, $(x,y,z)$ is a
subsequence of $\pi$ that is order-isomorphic to $(3,1,2)$. By Lemma
\ref{lem:Postorders312}, $\pi$ is not a postorder.
\end{proof}

While postorder insertion trees are less varied than for preorders, there may be many postorders with a given insertion tree. This affords some amount of freedom for choosing different sub-roots.

\begin{lemma}\label{lem:PostorderSubroots}
Let $\pi\setminus(3,1,2)=(p_1,\dots,p_n)$ be a postorder with insertion tree sequence $T_0,T_1,\dots,T_n$. For $1\le i\le n$, $p_i$ is either: \item
\begin{enumerate}[label=\upshape(\alph*)]
\item\label{item:NewHead} The \emph{single} sub-root greater than
$\max\{T_{i-1}\}$ (if present), or
\item\label{item:CombHead} The largest sub-root smaller than $\max\{T_{i-1}\}$ (if present).
\end{enumerate}
\end{lemma}

\begin{proof}
The result is vacuous for $i=1,2$. If $p_i$ is case \ref{item:NewHead}, we
merely note that if $p_i$ is a new maximum then it must be the right child of
the largest node in $\max\{T_{i-1}\}$. There can be at most one sub-root in
this position. Hence, $p_i$ is unique.

For the sake of contradiction, suppose $p_i$ is not of the form in case
\ref{item:NewHead} or \ref{item:CombHead}, and let $q$ be the largest sub-root
smaller than $\max\{T_{i-1}\}$. By Lemma \ref{lem:AncestorPrecedence}, the
items of each tooth are added in decreasing order. As $q$ is not the head of
its tooth, its successor $r$ must be in $T_{i-1}$, and furthermore $r$ precedes
both $p_i$ and $q$ in $\pi$. By construction, $(r,p_i,q)$ is a subsequence of
$\pi$. Yet this subsequence is isomorphic to $(3,1,2)$ since $p_i<q<r$,
contradicting Lemma \ref{lem:Postorders312}.
\end{proof}

\begin{figure}
\centering
\includegraphics[width=.85\textwidth]{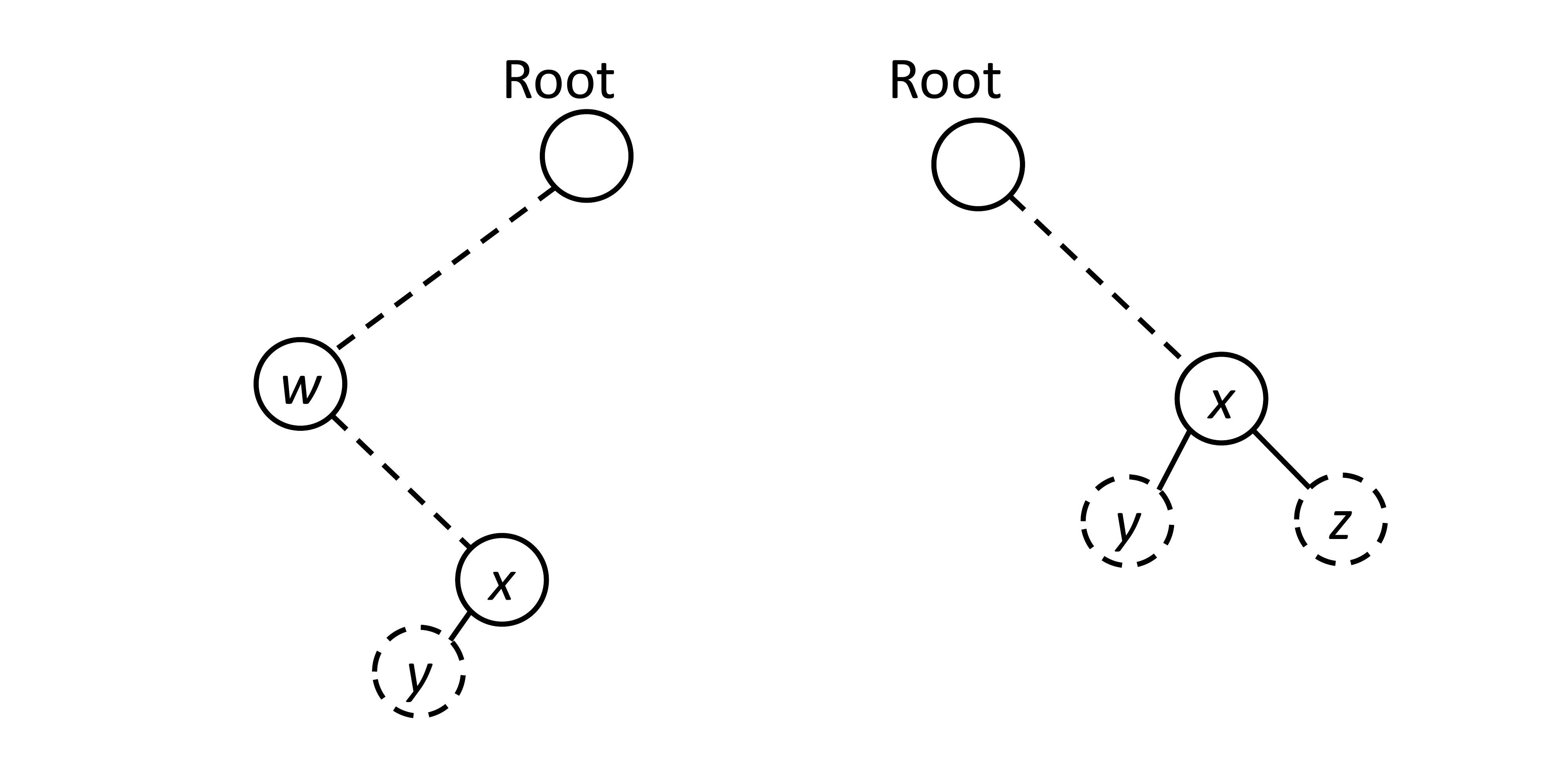}
\caption{Possible locations for the next sub-root $x$ to be insertion splayed
in $\pi\setminus(3,1,2)$. The case on the left occurs when the next splayed
node is less than the root, and the case on the right occurs when the next
sub-root is the new tree maximum. Dashed nodes may or may not be present, and
any number of nodes may lie on the paths denoted by dashed
lines.}\label{fig:PostorderSubroots}
\end{figure}

\begin{theorem}\label{thm:InsertionSplayingPostorders}
Insertion splaying postorders maintains the following invariants:
\begin{enumerate}
\item After each insertion splay, the path to every sub-root comprises $j\ge 0$
left pointers followed by $k\ge 0$ right pointers. (Furthermore, after the
first insertion, $k\ge 1$.)
\item The left-depth of every sub-root decreases from smallest to
largest.\footnote{The first two invariants dictate that the ancestors of
sub-roots form a \emph{right}-toothed comb.}
\item The splay operation takes constant amortized time.
\end{enumerate}
\end{theorem}

\begin{proof}
The base case is trivial. Lemma \ref{lem:PostorderSubroots} dictates that the
next splayed sub-root is either greater than all marked items, or is the
largest sub-root smaller than the tree root. Let $x$ be the next node to be
insertion splayed, $y$ its left child, and $z$ its right child (either or both
children may be missing).

Suppose $x$ is greater than the current tree root. Marking $x$ increases the
potential by $2$ and makes $y$ and $z$ new sub-roots. The splay operation
brings $x$ to the root by a sequence of left zig-zigs followed possibly by a
left zig (depending on whether the length of the access path is odd or even).
After each one of these zigs or zig-zigs, $y$'s left-depth remains $1$, and
$z$'s left depth remains $0$. Let $v$ be the root prior to the splay operation.
If the last splay step is a zig then the last splay operation increases the
left depth of $v$ and everything in its left subtree by either $1$ or $2$.
Since the left-depth of $x$ was $0$ and $x$ was the largest sub-root, the
inductive hypothesis ensures that all sub-roots had left-depth at least $1$
before the splay operation, and therefore at least $2$ afterward. Thus, when
$x$ becomes the root, the left-depths of each sub-root decrease from left to
right.

Otherwise, $x$ is the largest sub-root less than the root. Marking $x$ again
increases the potential by at most $2$. By Lemma
\ref{lem:PostorderInsertionTrees}, $x$ has no right child (see Figure
\ref{fig:PostorderSubroots}), so we only need to worry about its left child
$y$. Let $w$ be the last ancestor of $x$ that is a left child. Each left
zig-zig prior to the splay step involving $w$ maintains the left-depth of $y$
to be one greater than the left-depth of $x$. The splay step involving $w$ will
either be a left zig-zig or a left-right zig-zag, depending on the length of
the original path connecting $w$ to $x$. Regardless, immediately after the
splay step involving $w$, the ancestor of $y$ that is the left child of $x$ is
either the left child of $w$ or the left child of $w$'s parent. Since all the
sub-roots less than $y$ are in the left subtree of $w$, and thus have
left-depth greater than the left-depth of $y$, the invariant is restored, and
remains true after each right zig-zig or zig that brings $x$ to the root.

All that remains is showing constant amortized time. As noted before, marking
$x$ costs $2$. If $x$ is greater than the root then each left zig-zig, except
possibly the last, pays for itself, giving amortized cost of $4$. In the other
case, all splay steps except for the one involving $w$ and the one making $x$
the root pay for themselves, giving amortized cost at most $6$.
\end{proof}

\section{Balanced Trees}\label{sec:BalancedTrees}

Let $|x|$ denote the size of the subtree rooted at $x$. Following
\cite{WEIGHT_BALANCED_TREES}, we say $T$ is \emph{$\alpha$ weight balanced} for
$\alpha\in (0,1/2]$ if $\min\{|\Left(x)|, |\Right(x)|\}+1\ge
\alpha\cdot(|x|+1)$ for all $x\in T$, and write $T\in\BB[\alpha]$.

\begin{theorem}\label{thm:BalancedTraversal}
For any \emph{(fixed)} $0<\alpha\le 1/2$, if $S\in\BB[\alpha]$ and $T$ has the
same keys as $S$, then the cost of splaying $\Preorder(S)$ or $\Postorder(S)$
starting from $T$ is $O(|T|)$.
\end{theorem}

\begin{proof}
By Theorem \ref{thm:DynamicFinger}, it suffices to show that
$\DF_T(\Preorder(S))=O(|T|)$. Let
\begin{align*}
A_\alpha(n)\equiv\max\{\DF_T(\Preorder(S))\mid\text{$S\in\BB[\alpha]$ and $|T|=n$}\}.
\end{align*}

Recall that $\Preorder(S)=(\Root(S))\oplus\Preorder(L)\oplus\Preorder(R)$,
where $L$ and $R$ are the left and right subtrees of the root of $S$,
respectively. Notice that the rank differences between $\Root(S)$ and the first
item in $\Preorder(L)$, and between the last item in $\Preorder(L)$ and the
first item in $\Preorder(R)$, are at most $|T|$ by definition. Hence,
\begin{align*}
\DF_T(\Preorder(S)) \le \DF_T(\Preorder(L)) + \DF_T(\Preorder(R)) + 2\log_2(|T|+1).
\end{align*}

Observe that $(|L|+1)/(|S|+1)\in[\alpha,1-\alpha]$ since $S\in\BB[\alpha]$, and by definition $|R|<|S|-|L|$. Hence,
\begin{align*}
A_\alpha(n)=\max_{\alpha\le\beta\le1/2}\{A_\alpha(\beta\cdot n) + A_\alpha((1-\beta)\cdot n)\} + O(\log n).\footnotemark
\end{align*}
\footnotetext{Technically, since $|L|/|S|<(|L|+1)/(|S|+1)$, we need to pick $S$ sufficiently large for a given alpha, and offset the recurrence term by a corresponding constant. This does not asymptotically affect the result.}%
Akra-Bazzi's result \cite{AKRA_BAZZI} suffices to show $A_\alpha(n)=O(n)$ for
fixed $\alpha$. The proof for postorders is identical.
\end{proof}

\begin{remark}
In actuality, $A_\alpha(n)=O(f(\alpha)\cdot n)$ for some function $f$ of
$\alpha$. Unfortunately, the computation appears to be messy. We have declined
to do the necessary footwork, as we strongly suspect that, regardless,
$A_\alpha(n)$ does not tightly bound the cost of splaying these sequences.
\end{remark}

\begin{remark}
This result extends to any binary search tree algorithm that satisfies the
dynamic finger bound. Iacono and Langerman proved Greedy also has the dynamic
finger property \cite{GREEDY_DYNAMIC_FINGER}; their analysis does not
consider initial trees, however.
\end{remark}

\section{Remarks}\label{sec:GeneralPatterns}

Patterns that avoid $(2,1,3)$ are ``symmetric'' to those that avoid $(2,3,1)$:
if $\pi\setminus(2,1,3)$ then $\pi$ is the preorder of the mirror image of
$\BST(\pi)$. Similarly, patterns that avoid $(1,3,2)$ are symmetric to patterns
that avoid $(3,1,2)$. Thus, insertion splaying $\pi\setminus(2,1,3)$ and
$\pi\setminus(1,3,2)$ takes linear time.

The only other patterns of length three are $(3,2,1)$ and its symmetric
counterpart $(1,2,3)$. The pattern $(3,2,1)$ was explored in
\cite{PATTERN_AVOIDANCE}, where it was shown that Greedy executes
$(3,2,1)$-avoiding permutations in linear time starting from an
\emph{arbitrary} tree. In fact, they showed that executing
$\pi\setminus(k,\dots,2,1)$ takes time proportional to $n\cdot 2^{O(k^2)}$;
this is linear in $n$ for fixed $k$. These permutations are called
\emph{$k$-increasing} because they can be partitioned into $k-1$ disjoint
monotonically increasing subsequences \cite{PATTERN_AVOIDANCE}. They form the
natural generalization of \emph{sequential access}, which is the (unique)
permutation of the tree nodes that avoids $(2,1)$.

More general invariants can be derived about insertion tree structure and
sub-root insertion order based on pattern-avoidance. As one particularly
interesting example:

\begin{theorem}
If $\pi\setminus(k,\dots,2,1)$ then no node in $\BST(\pi)$ has left-depth more
than $k-2$, and the next sub-root inserted (without splaying) is always the
smallest sub-root with its given left-depth.
\end{theorem}

The proof is similar to Lemmas \ref{lem:PostorderInsertionTrees} and
\ref{lem:PostorderSubroots}. In particular, the insertion trees of
$(3,2,1)$-avoiding permutations look like the combs of postorder insertion
trees, except the teeth are rightward, instead of leftward paths.

For $k$-increasing sequences, the potential used for Theorems
\ref{thm:InsertionSplayingPreorders} and \ref{thm:InsertionSplayingPostorders}
needs modifications. The main issue is that in both of these cases, the
zig-zigs paid for themselves because the nodes knocked off the access path did
not have sub-root descendants. This structure no longer holds for
$(3,2,1)$-avoiding sequences, since we must splay the nodes of the teeth in
increasing order. The proof seems to require a generalization of the sequential
access theorem. It is possible that the notion of \emph{kernel trees} used by
Sundar in \cite{SUNDAR_DEQUE} for a potential-based proof of the sequential
access theorem could be useful.

\printbibliography

\end{document}